\documentclass[12pt]{amsart}

\usepackage{amsmath,amsthm,amssymb,amsfonts,verbatim,multirow}
\usepackage[colorlinks,
            linkcolor=blue,
            anchorcolor=blue,
            citecolor=blue
            ]{hyperref}
\usepackage[hmargin=1.2in,vmargin=1.2in]{geometry}

\title[Maximally recoverable LRCs from subspace direct sum systems]{Maximally recoverable local reconstruction codes from {subspace direct sum systems}}
\author{Shu Liu}\address{National Key Laboratory of Science and Technology on Communications, University of Electronic Science and Technology of China, Chengdu, China} \email{shuliu@uestc.edu.cn}
\author{Chaoping Xing} \address{School of Electronic Information and Electric Engineering, Shanghai Jiao Tong University, Shanghai, China} \email{xingcp@sjtu.edu.cn}
\date{}

\date{}

\newtheorem{lemma}{Lemma}[section]
\newtheorem{theorem}[lemma]{Theorem}
\newtheorem{cor}[lemma]{Corollary}

\newtheorem{defn}{Definition}

\theoremstyle{remark}
\newtheorem{rmk}{Remark}

\renewcommand{\epsilon}{\varepsilon}
\renewcommand{\le}{\leqslant}
\renewcommand{\ge}{\geqslant}

\newcommand{\vnote}[1]{}

\def\F{\mathbb{F}}
\def \mC {\mathcal{C}}

\def \mB {\mathcal{B}}
\def \mC {\mathcal{C}}

\def \mL {\mathcal{L}}

\def \mP {\mathcal{P}}

\def \mX {\mathcal{X}}

\def\Pin{{P_{\infty}}}

\def \Xi {{X^{[i]}}}

\newcommand{\Ga}{\alpha}
\newcommand{\Gb}{\beta}
     
\newcommand{\Gd}{\delta}

\newcommand{\Gl}{\lambda}

\def \bc {{\bf c}}

\def \bu {{\bf u}}
\def \bv {{\bf v}}
\def \bo {{\bf 0}}

\def\SP{{\mathbb{S}}}
\def\wt{{\rm wt}}

\begin{document}

\maketitle

\begin{abstract} Maximally recoverable local reconstruction codes (MR LRCs for short) have received great attention in the last few years. Various constructions have been proposed in literatures. The main focus of this topic is to construct MR LRCs over small fields. An $(N=nr,r,h,\Gd)$-MR LRC is a linear code over finite field $\F_\ell$ of length $N$, whose
codeword symbols are partitioned into $n$ local groups each of size $r$. Each local group can repair any $\Gd$ erasure errors and there are further
$h$ global parity checks to provide fault tolerance from more global erasure patterns.

MR LRCs deployed in practice have a small number of global parities such as $h=O(1)$. In this parameter setting, all previous constructions require the field size $\ell =\Omega_h (N^{h-1-o(1)})$. It remains challenging to improve this  bound. In this paper, via subspace direct sum systems, we present a construction of MR LRC with  the field size $\ell= O(N^{h-2+\frac1{h-1}-o(1)})$. In particular, for the most interesting cases where $h=2,3$, we improve previous constructions by either reducing field size or removing constraints. In addition, we also offer some constructions of MR LRCs for larger global parity $h$ that have field size incomparable with known upper bounds. The main techniques used in this paper is through subspace direct sum systems that we introduce. Interestingly, subspace direct sum systems are actually equivalent to $\F_q$-linear codes over extension fields. Based on various constructions of  subspace direct sum systems, we are able to construct several classes of MR LRCs.
\end{abstract}

\section{Introduction}
For a distributed storage system, data is partitioned and stored in different  servers as each serve has only a small storage capacity of a few terabyte. A server may  crash or become temporarily unavailable due to system updates/network bottlenecks. Thus, we have to design well-structured architectures to overcome these two problems. In the case of crash, we have to repair the lost data stored in the crashed servers. In the second case of temporary unavailability, we have to  service user requests with low latency despite some servers becoming temporarily
unavailable. Instead of just replicating data which is wasteful, distributed storage systems use
erasure codes. In view of this demand, Local Reconstruction Codes (LRCs for short)
were invented \cite{HCL07,CHL07} precisely for achieving the objectives of local computations while still maintaining storage efficiency
and have been implemented in several large scale systems such as Microsoft Azure \cite{HSX12} and
Hadoop \cite{SAP13}. These codes  can recover quickly from a
small number of erasures by reading only a small number of available servers; and at the same
time, they can also recover from the unlikely event of a large number of erasures (but can do so
less efficiently). Locality in distributed storage was first introduced in \cite{HCL07,CHL07}, but locally reconstruction codes
were first formally defined and studied in \cite{GHSY12} and \cite{PD14}.

The current paper concerns a much stronger requirement on global fault-tolerance, called {\em Maximal Recoverability} (MR for short). This requires that the code should simultaneously correct every erasure pattern that is information-theoretically possible to correct, given the locality conditions imposed on the codeword symbols. Let us describe it more formally. Define an $(N,r,h,\Gd)_\ell$-LRC to be a linear code over $\F_\ell$ of length $N$
whose $N$ codeword symbols are partitioned into $n$ disjoint groups each of which include $\Gd$ local parity checks capable of locally correcting $\Gd$ erasures. The codeword symbols further obey $h$ heavy (global) parity checks. With this structure of parity checks, it is not hard to see that the erasure patterns one can hope to correct are precisely those which consist of up to $\Gd$ erasures per local group plus up to $h$ additional erasures anywhere in the codeword. An MR LRC is a \emph{single} code that is capable of simultaneously correcting \emph{all} such patterns.

Since encoding a linear code and decoding it from erasures involve performing numerous finite field arithmetic operations, it is highly  desirable to have codes over small fields (preferably of characteristic 2). Obtaining MR LRCs over finite fields of minimal size has therefore emerged as a central problem in the area of codes for distributed storage.

\subsection{Known results}
The  maximally recoverable LRCs were introduced in \cite{BHH13} and motivated by applications to storage on solid-state devices, where they were called partial MDS codes. The terminology ``maximally recoverable codes" was coined in \cite{GHJY14}, and the concept was more systematically studied in \cite{GHJY14,GHK17}.

First of all,  a lower bound on the field size was presented in \cite{GGY20}. Stating the bound when $h \le \frac{N}{r}$ for simplicity, they show that the field size $\ell$ of an $(N=nr,r,h,\Gd)_\ell$ MR LRC must obey
\begin{equation}\label{eq:1}
\ell=\Omega_{\Gd,h}\left(N\cdot r^{\min\{\Gd,h-2\}}\right).
\end{equation}
 The lower bound \eqref{eq:1} is still quite far from the currently best-known upper bounds. In particular, the exponent of $\Gd$ or $h$ is to the base growing with $n$ in the known constructions, but only to the base $r$ in the above lower bound. Thus, one can conjecture that there is still room to improve  the lower bounds.

There are several known constructions of MR LRCs which are incomparable to
each other in terms of the field size \cite{GG21,CMST21,GHJY14,GYBS19,GJX20,MK19,GGY20,Bla13,TPD16,
HY16,BPSY16}. Some constructions are better than others based on the range
of parameters. As it is difficult to list all known constructions one by one, let us summarize the parameters of the MR LRCs that have been constructed so far  in Table I in the next subsection.

MR LRCs deployed in practice have a small number of global parities, typically $h = 2, 3$ \cite{HSX12}. For $h=2$, MR LRCs with linear field size were constructed for fields of odd characteristic. If a field has the even characteristic, to get a linear   field size, one requires that $N=\Theta(r^2)$ \cite{HSX12}. For $h=3$, the smallest field size is $\ell=O(N^{3/2})$ if $r$ is constant and $\ell$ is even \cite{GHJY14}. For other cases, the best known result is that the field size is $O(n^3)$ \cite{CMST21}.

Let us summarize some of recent constructions in \cite{GG21,CMST21,Mar20}.

Via skew polynomials,  Gopi and Guruswami \cite{GG21} presented a construction with the field size
\begin{equation}\label{eq:2}
\ell=O\left(\max\{r,N/r\}^{\min\{h,r-\Gd\}}\right).
\end{equation}
Prior to the work \cite{GG21}, {Cai et. al.} \cite{CMST21} used the similar idea to construct MR LRCs with field size $\ell=O\left(\max\{r,N/r\}^h\right)$.

Soon after the work of {\cite{CMST21}}, there are two constructions given in \cite{Mar20} with field sizes
\begin{equation}\label{eq:3}
\ell=O\left(\max\{(2r)^{r-\Gd},N/r^2\}^{\min\{h,\lfloor N/r^2\rfloor\}}\right)
\end{equation}
and
\begin{equation}\label{eq:4}
\ell=O\left((2r)^{r-\Gd}(\lfloor N/r^2\rfloor+1)^{h-1}\right).
\end{equation}

\subsection{Our results and comparison} A particularly interesting setting of parameters is $h=O(1)$ and $r=N^{o(1)}$. This setting was specifically mentioned in \cite{GG21}. Let us quote a sentence about this setting of parameters from \cite{GG21}.
 ``Despite all these constructions, a particularly interesting setting of parameters, which remains
challenging is the case when $h = O(1)$ and $r = N^{o(1)}$. The lower bound \eqref{eq:1} only shows that the field size
$\ell = \Omega_h(N^{1+o(1)})$ whereas all the existing constructions require $\ell =\Omega_h (N^{h-1-o(1)})$." One of the main results of this paper is improvement of this upper bound to $\ell =O( N^{h-2+{1}/{(h-1)}+o(1)})$ (see Table I below). Thus, we beat all known constructions for this parameter setting. In particular, for $h=2$, we reduce the field size to $\ell=O(N)$  and $\ell=O(N^{1+o(1)})$ for constant $r$ and $r=o{(\log N/\log\log N)}$, respectively. Note that for even $\ell$, the construction given in \cite{GGY20}  gives linear size $\ell=\Theta(N)$ subject to the constraint $N=\Theta(r^2)$. For odd $\ell$, the field size is $\ell=\Theta(N)$ in \cite{GGY20}.

 For $h=3$, we obtain the field size $\ell=O(N^{3/2})$ and $\ell=O(N^{3/2+o(1)})$ for constant $r$ and $r=o({\log N/\log\log N})$, respectively. Note that in \cite{GHJY14}, one requires $\ell$ to be even in order to
have the field size $\ell=O(N^{3/2})$. In addition, this result improves all known constructions when $r=o(\log N/\log\log N)$. One particular case that we want to mention is $h=5$. From the upper bound mentioned above, we have $\ell =O( N^{13/4+o(1)})$ for $h=5$. However, from the cyclic codes constructed in \cite{D95}, we obtain the field size $\ell =O( N^{3+o(1)})$ for $h=5$. This gives the best known field size for $h=5$ as far as we know.

For large $h$, we have a few upper bounds that are incomparable to  known ones in terms of the field size. For instance, for our bound 
$\ell=O((2r)^{hr})$ with $r=\Omega(\log n/\log \log n)$, it is better than the one given in {\cite[Section IV.A]{MK19}} under the regime $n>(2r)^{\frac{rh}{r-\delta}}$ (note that in general, $\frac{r}{r-\delta}$ is a constant, it implies that the regime $n>(2r)^{\frac{rh}{r-\delta}}$ is the same as $h<O({\log n}/{\log r})$); while it outperforms all other bounds (for arbitrary $h$) given in Table I. 
The second bound that we derive in this paper is $\ell=O\left((2r)^{h+\frac{n}{r^{r/2}-1}}\right)$ when $r=O(\log n/\log\log n)$. This bound is reduced to $\ell=O\left((2r)^{h(1+o(1))}\right)$ when $n/r^{r/2}=o(h)$, i.e., $n=o(hr^{r/2})$.
Compared with the bound given {\cite[Section IV.A]{MK19}} (see Table I below), our bound is better for the regime $h<r\cdot\frac{\log n}{\log r}$. Moreover, this second bound
 beats all other previously known bounds when $n=o(hr^{r/2})$. We also have some other bounds that are incomparable with previous known bounds given in Table I. We list them in Table I as well.

{\footnotesize

\begin{center}~\label{table:1}

Table I\\  $(N=nr, r, h, \delta)$-MR codes \\ \smallskip

{\rm

\begin{tabular}{|c|c|c|cl}\hline\hline

$h$ & Field size $\ell$ & Restrictions & \multicolumn{1}{c|}{References}\\ \hline


{\multirow{6}{*}{$2$}}  & $\Theta(N\delta)$ & -- &  \multicolumn{1}{c|}{\cite[Theorem 7]{BPSY16}} \\\cline{2-4}

&$\Theta(N)$ & $\ell$ is odd &\multicolumn{1}{c|}{\cite[Theorem IV.4]{GGY20}}  \\ \cline{2-4}

 &$N\cdot\exp(O(\sqrt{\log_q N}))$ &  $\ell$ is even & \multicolumn{1}{c|}{\cite[Theorem IV.4]{GGY20}}\\\cline{2-4}

& $\Theta(N)$& $\ell$ is even, $N=\Theta(r^2)$ & \multicolumn{1}{c|}{\cite[Theorem IV.4]{GGY20}}\\ \cline{2-4}

& $O(N)$& $r$ is a constant & \multicolumn{1}{c|}{Theorem 4.3(i)}\\  \cline{2-4}

& $O(N^{1+o(1)})$& $r=o(\log n/\log \log n)$  & \multicolumn{1}{c|}{Theorem 4.3(ii)}\\  \hline


\multirow{4}{*}{$3$}  & $O(N^{3/2})$ & $r$ is a constant, $\ell$ is even &  \multicolumn{1}{c|}{\cite[Corollary 23]{GHJY14}} \\ \cline{2-4}

   & {$O((N/r)^3)$} & {$r=O(\sqrt{N})$} &  \multicolumn{1}{c|}{\cite[Construction A]{CMST21}}\\\cline{2-4}

  & {$(O(N/r))^3$} & {$n+1$ and $r-1$ are prime powers} &  \multicolumn{1}{c|}{\cite[Theorem 1.3]{GG21}}\\\cline{2-4}

 & $O(N^{3/2})$ &$r$ is a constant  &  \multicolumn{1}{c|}{Theorem 4.3 (i)} \\ \cline{2-4}

 & $O(N^{3/2+o(1)})$& $r=o({\log n/\log\log n})$  &  \multicolumn{1}{c|}{Theorem 4.3 (ii)} \\   \hline


\multirow{2}{*}{$5$}  & $O(N^{3})$ &$r$ is a constant  &  \multicolumn{1}{c|}{Theorem 4.3 (iii)} \\ \cline{2-4}

 & $O(N^{3+o(1)})$& $r=o({\log n/\log\log n})$  &  \multicolumn{1}{c|}{Theorem 4.3 (iii)} \\   \hline


\multirow{2}{*}{$O(1)$}  & $O(N^{h-2+{\frac{1}{h-1}}})$ &$r$ is a constant  &  \multicolumn{1}{c|}{Theorem 4.3 (i)} \\\cline{2-4}

 & $O(N^{h-2+{\frac{1}{h-1}}+o(1)})$& $r=o({\log n/\log\log n})$  &  \multicolumn{1}{c|}{Theorem 4.3 (ii)} \\   \hline


 \multirow{16}{*}{Arbitrary}  & $\Theta(r^{N(r-\delta)/r})$ &--  &  \multicolumn{1}{c|}{\cite[Corollary 11]{CK16}} \\ \cline{2-4}

&  \multirow{2}{*}{$\Theta(rN^{h(\delta+1)-1})$} & {$r$ is a prime power,} &  \multicolumn{1}{c|}{\multirow{2}{*}{\cite[Lemma 7]{GYBS19}}} \\

& & $2n$ is a power of $r$&\multicolumn{1}{c|}{} \\ \cline{2-4}

 & \multirow{2}{*}{$\Theta(\max(N/r,r^{h+\delta})^h)$} &{$r$ is a prime power,}&  \multicolumn{1}{c|}{\multirow{2}{*}{\cite[Corollary 10]{GYBS19}}} \\

&&${N}/{r}+{1}$ is a power of $r$&\multicolumn{1}{c|}{}\\\cline{2-4}

& $\Theta(\max(N/r,(2r)^{h+\delta})^{\min(N/r,h)})$ &$r=o({\log n/\log\log n})$&  \multicolumn{1}{c|}{\cite[Theorem 17]{GJX20}} \\   \cline{2-4}

& $\Theta(\max(N/r,(2r)^{r})^{\min(N/r,h)})$ &$r=o({\log n/\log\log n})$ &  \multicolumn{1}{c|}{\cite[Theorem 19]{GJX20}} \\ \cline{2-4}

& $\Theta(\max(r,N/r)^{r-\delta})$ &-- &  \multicolumn{1}{c|}{\cite[Section IV.A]{MK19}} \\\cline{2-4}

& {$\Theta(2^h(\max(r,N/r))^{h})$} & $\delta\ge\lfloor{\frac{h}{r-\delta}}\rfloor+1$ &  \multicolumn{1}{c|}{\cite[Construction A]{CMST21}} \\ \cline{2-4}




& $O((2\max\{r,n\})^{\min\{h,r-\delta\}})$ & $h<n(r-\delta)$& \multicolumn{1}{c|}{\cite[Theorem 1.3]{GG21}} \\\cline{2-4}

& $(\max\{(2r)^{r-\delta},N/r^2\})^{\min\{h,N/r^2\}}$ & $r=\Theta(\log n/\log\log n)$& \multicolumn{1}{c|}{\cite[Corollary 40]{Mar20}} \\\cline{2-4}

& $(2r)^{r-\delta}(\lfloor{N}/{r^2}\rfloor+1)^{h-1}$ & $r=\Theta(\log n/\log\log n)$& \multicolumn{1}{c|}{\cite[Corollary 43]{Mar20}} \\\cline{2-4}


& $O((2r)^{hr})({N}/{r})^{h-1}$ & --& \multicolumn{1}{c|}{Theorem 4.4 (i)} \\\cline{2-4}

& $O((2r)^{hr})$ & $r=\Omega(\log n/\log\log n)$& \multicolumn{1}{c|}{Theorem 4.4 (iii)} \\\cline{2-4}

& $O\left((2r)^{h+\frac{n}{r^{r/2}-1}}\right)$ & $r=O(\log n/\log\log n)$& \multicolumn{1}{c|}{Theorem 4.4 (v)} \\\cline{2-4}

& $O\left((2r)^{h(h+\Gd)}\right)$ & $h+\Gd<r$ and $h+\Gd=O(\log n/\log r)$ & \multicolumn{1}{c|}{Theorem 4.7 (i)}
\\ \hline\hline

\end{tabular}

}

\end{center}

}

``--" in the above table means that there is no restriction.

\subsection{Our idea and techniques}
The idea is to construct a proper parity-check matrix over $\F_\ell$. To have a required parity-check matrix, we divide columns of the bottom matrix of a parity-check into $n$ blocks $D_1,D_2,\dots,D_n$ (see \eqref{eq:8}). Let $\F_q$ be a subfield of $\F_\ell$ with $\ell=q^m$. We define $D_i$ to be a $(q,\ell)$-Moore matrix. Let $V_i$ be the $\F_q$-subspace spanned by the first row of the $i$th block $D_i$. To obtain the desired MR LRCs, we require that any $h$ subspaces out of $V_1,V_2,\dots,V_n$ form a direct sum. To have small field size $\ell$, we want that the elements of $D_i$ belong to a small finite field $\F_\ell$, i.e., small $m$. Thus, we define a subspace direct sum system to be a set $\{V_i\}_{i=1}^n$ with each $V_i\subseteq \F_{q^m}$ such that any $h$ subspaces  form a direct sum. Hence, to have an MR LRC over small field $\F_\ell$, we require a subspace direct sum system with small $m$.

To get constructions and some bounds on subspace direct sum systems, we convert subspace direct sum systems to $\F_q$-linear block codes over extension fields with large block Hamming weight. We then explore various constructions of such $\F_q$-linear block codes over extension fields to obtain subspace direct sum systems with reasonable parameters. Finally, we present a construction of MR LRCs via subspace direct sum systems. Based on those subspace direct sum systems that we obtain in this paper, we give a few classes of MR LRCs.

\subsection{Organization of the paper}
The paper is organized as follows. In Section 2, we present some preliminaries including definition of MR LRCs, their generator and parity-check matrices, Moore matrices, algebraic geometry codes, etc. In Section 3, we introduce subspace direct sum system and then show that such a system is equivalent to a linear block code. We then derive some upper and lower bounds on subspace direct sum systems. We provide some constructions of subspace direct sum systems in Section 3 as well. In the last section, we show how to make use of subspace direct sum systems to construct MR LRCs and then construct MR LRCs over small field through subspace direct sum systems provided in Section 3.

\section{Preliminary}
In this section, we introduce some notations and basic results that will be used in this paper. For a prime power $q$, denote by $\F_q$ the finite field of $q$ elements. For a positive integer $n$, we denote by $[n]$ the set $\{1,2,\dots,n\}$. We denote by $\log$ the logarithm with base $2$.
\subsection{Maximally recoverable codes} Let $\ell$ be a prime power.
An $\ell$-ary linear code $\mC$ of length $n$ is an $\F_\ell$-subspace of $\F_\ell^n$. We denote by $[N,k,d]_q$ a $q$-ary linear code with length $N$, dimension $k$ and minimum distance $d$. Then the classical Singleton bound says that an $\ell$-ary $[N,k,d]$-linear code must obey
\begin{equation}\label{eq:5}
k+d\le N+1.
\end{equation}
Now we show that linear independence of columns of a  matrix remains unchanged regardless of row operations.
\begin{lemma}\label{lem:2.1}
Let $H=(\bu_1,\bu_2,\dots,\bu_N)$ be a matrix in $\F_\ell^{m\times N}$. Let $H'=(\bu'_1,\bu'_2,\dots,\bu'_N)$ be a matrix obtained from $H$ by elementary row operations. Then for a subset $I\subseteq [N]$, the vectors $\{\bu_i\}_{i\in I}$ are linearly independent if and only if $\{\bu'_i\}_{i\in I}$ are linearly independent.
\end{lemma}
\begin{proof} Assume that $\{\bu_i\}_{i\in I}$ are linearly independent.
As $H'$ is obtained from $H$ by elementary row operations, there exists an invertible matrix $A\in\F_\ell^{m\times m}$ such that $AH=H'$, i.e., $\bu'_i=A\bu_i$. Hence, $\{\bu_i\}_{i\in I}$ are linearly independent if and only if $\{\bu'_i=A\bu_i\}_{i\in I}$ are linearly independent. Conversely, it is also true.
\end{proof}

Now we start to introduce definition of locally repairable codes. In this paper, we consider disjoint recovery sets only. Hence, we divide $[N]$ into $n$ groups $R_1,\dots,R_n$ with each group $R_i$ having $r$ elements. Thus, we assume that $N=nr$. For an integer $\Gd$ with $1\le \Gd\le r-2$, we say that  an $\ell$-ary $[N,k,d]$-linear code $\mC$ has $(r-\Gd, \Gd+1)$-locality if the projection $\mC|_{R_i}$ of $\mC$ at each $R_i$ is a code of distance at least $\Gd+1$, {i.e., it can be locally repaired by connecting to $r-\delta$ other nodes}. The Singleton-type  $[N,k,d]$-locally repairable  code with locality  $(r-\Gd, \Gd+1)$ says that
\begin{equation}\label{eq:6}
d\le n-k+1-\left(\left\lceil\frac{k}{r-\Gd}\right\rceil-1\right)\Gd.
\end{equation}
An $\ell$-ary $[N,k,d]$-locally repairable  code with locality  $(r-\Gd, \Gd+1)$ is said Singleton-optimal if the bound \eqref{eq:6} is achieved.

Informally, an $(N,r,h,{\Gd})_\ell$ maximally recoverable code is an $\ell$-ary linear code such that erasure errors at any ${\Gd}$ positions in each group and other $h$ position in arbitrary positions can be corrected.

Now we give a formal definition of maximally recoverable codes.
\begin{defn}\label{def:1} Let $\ell$ be a prime power. Let $R_i$ be a partition of $[N]$ with each $R_i$ having $r$ elements.
An $(N=nr,r,h,{\Gd})_\ell$ maximally recoverable code (or $(N=nr,r,h,{\Gd})_\ell$-MR code for short) $\mC$ is an $[N,k]$-linear code over $\F_\ell$ with $k=N-n\Gd-h$ satisfying
\begin{itemize}
\item[(i)] the projection $\mC|_{R_i}$ of $\mC$ at each $R_i$  is an $[r,r-{\Gd},{\Gd}+1]_\ell$-MDS code; and
\item[(ii)] Puncturing any ${\Gd}$ positions in each $R_i$  gives an $[N-n\Gd,k,h+1]$-MDS code.
\end{itemize}
\end{defn}
From Definition \ref{def:1}, it is well known  that generator and parity-check matrices of an $(N,r,h,{\Gd})_\ell$-MR code must have the following forms (after re-arrangement of columns). For the sake of completeness, we provide a proof below.
\begin{lemma} ~\label{lem:2.2} $\mC$ is an $(N=nr,r,h, {\Gd})_\ell$-MR code if and only if
\begin{itemize}
\item[{\rm (i)}]  it has a generator matrix of the form
\begin{equation}\label{eq:7}
G=(B_1|B_2|\cdots|B_n)\in\F_q^{k\times N},
\end{equation}
where $B_i$ is a $k\times r$ matrix and generates an $[r,r-{\Gd},{\Gd}+1]_\ell$-MDS code. Furthermore, deleting any ${\Gd}$ columns in each $B_i$ from $G$ gives a generator matrix of an $[N-n\Gd,k,h+1]$-MDS code with $k=N-n\Gd-h$; and
\item[{\rm (ii)}]  it has a parity-check matrix of the form
\begin{equation}\label{eq:8} H=\left(\begin{array}{c|c|c|c}
A_1&O&\cdots&O\\ \hline
O&A_2&\cdots&O \\ \hline
\vdots&\vdots&\ddots&\vdots \\ \hline
O&O&\cdots&A_n \\ \hline
D_1&D_2&\cdots&D_n
\end{array}
\right)\in\F_{\ell}^{(N-k)\times N},\end{equation}
 where each $A_i$ has size ${\Gd}\times r$ and each $D_i$ has size $h\times r$ and they satisfy:
(a) each $A_i$  is a generator matrix of an $[r,{\Gd},r-{\Gd}+1]_\ell$-MDS code for $1\le i\le n$; and
(b) every $n\Gd+h$  columns consisting of any ${\Gd}$ columns in each $R_i$ and other arbitrary $h$ columns are $\F_\ell$-linearly independent.
\end{itemize}
\end{lemma}
\begin{proof} Part (i) is clear.

Let us prove (ii). Assume that $\mC$ is an $[N-n\Gd,k,h+1]_\ell$-MDS code with $k=N-n\Gd-h$. Let $G=(B_1|B_2|\cdots|B_n)\in\F_q^{k\times N}$ be a generated matrix of $\mC$.

Let $A_i$ be the parity-check matrix of the code spanned by $B_i$. Then $A_i$ is a generator  matrix of an $[r,{\Gd},r-{\Gd}+1]_\ell$-MDS code. Hence, the matrix
\begin{equation}\label{eq:9} H_1=\left(\begin{array}{c|c|c|c}
A_1&O&\cdots&O\\ \hline
O&A_2&\cdots&O \\ \hline
\vdots&\vdots&\ddots&\vdots \\ \hline
O&O&\cdots&A_n
\end{array}
\right)\in\F_{\ell}^{(N-k-h)\times N},\end{equation}
has rank $n\Gd=N-k-h$. Moreover, every row of the matrix $H_1$ in \eqref{eq:9} is a codeword in $\mC^\perp$. Thus, we can extend the rows of $H_1$ into a basis of $\mC^\perp$ by adding $h$ codewords. Now by putting these $h$ codewords at the bottom of $H_1$, we get a parity-check matrix $H$ of $\mC$ in the form \eqref{eq:8}. Next we shall prove that $H$ satisfies both (a) and (b) in (ii). By choices of $A_i$, it is clear that (a) is fulfilled. Now we take any ${\Gd}$ columns in each $R_i$ and re-arrange  columns of $H$. Then $H$ can be written into the following form
\[ H=\left(\begin{array}{c|c|c|c}
C_1,E_1&O&\cdots&O\\ \hline
O&C_2,E_2&\cdots&O \\ \hline
\vdots&\vdots&\ddots&\vdots \\ \hline
O&O&\cdots&C_n,E_n\\ \hline
L_1,K_1&L_2,K_2&\dots&L_n,K_n
\end{array}
\right)\in\F_{\ell}^{(N-k)\times N},\]
where $C_i, E_i, L_i$ and $ K_i$ are ${\Gd}\times {\Gd}$, ${\Gd}\times (r-{\Gd})$, $h\times {\Gd}$ and $h\times (r-{\Gd})$ matrices, respectively, with the first ${\Gd}$ columns in each block were previously chosen. Since $A_i$ is a generator matrix of an MDS code, $A_i$ is a ${\Gd}\times {\Gd}$ invertible matrix for all $1\le i\le n$. By multiplying an invertible matrix from the left of $H$, we obtain
\[
H_2=\left(\begin{array}{c|c|c|c}
C_1,E_1&O&\cdots&O\\ \hline
O&C_2,E_2&\cdots&O \\ \hline
\vdots&\vdots&\ddots&\vdots \\ \hline
O&O&\cdots&C_n,E_n\\ \hline
O,K_1-L_1C_1^{-1}E_1&O,K_2-L_2C_2^{-1}E_2&\dots&O,K_n-L_nC_n^{-1}E_n
\end{array}
\right),\]
As $H_2$ is still a parity-check matrix of $\mC$, the matrix $(O,K_1-L_1C_1^{-1}E_1|O,K_2-L_2C_2^{-1}E_2|\dots|O,K_n-L_nC_n^{-1}E_n)$ has rank $h$, i.e.,
\[H_3:=(K_1-L_1C_1^{-1}E_1|K_2-L_2C_2^{-1}E_2|\dots|K_n-L_nC_n^{-1}E_n)\]
has rank $h$. We can see that $H_3$ is in fact a parity-check matrix of the code obtained from $\mC$ by puncturing  ${\Gd}$ positions in each $R_i$. This implies that any $h$ columns of $H_3$ are linearly independent.  Hence, the first ${\Gd}$ columns in $H_2$ together with any other $h$ columns of $H_2$ are linearly independent. As $H_2$ is obtained from $H$ by elementary row operations, by Lemma \ref{lem:2.1}, we conclude that (b) in part (ii) holds.

If $\mC$ has a parity-check matrix of the form \eqref{eq:8} satisfying (a) and (b) in (ii), by using a similar argument given above, we can show that puncturing any ${\Gd}$ positions in each $R_i$  gives an $[N-n\Gd,k,h+1]$-MDS code.
This completes the proof.
\end{proof}

  \subsection{Moore matrix}
  \label{subsec:2.2}
  Let $\ell$ be a power of $q$. For elements $\Ga_1,\dots,\Ga_n\in\F_\ell$, a $(q,\ell)$-Moore matrix of size $h\times n$ is defined by
  \[M(\Ga_1,\dots,\Ga_n)=\left(\begin{array}{cccc}
\Ga_1&\Ga_2&\cdots&\Ga_n\\
\Ga_1^q&\Ga_2^q&\cdots&\Ga_n^q\\
\vdots&\vdots&\ddots&\vdots \\
\Ga_1^{q^{h-1}} &\Ga_2^{q^{h-1}} &\cdots&\Ga_n^{q^{h-1}}
\end{array}
\right)\in\F_{\ell}^{h\times n}.\]
If $n=h$, then $M(\Ga_1,\dots,\Ga_h)$ is a square matrix. In this case,
  the determinant $\det(M(\Ga_1,\dots,\Ga_h))$ is given by the following formula
  \[\det(M(\Ga_1,\dots,\Ga_h))=\prod_{(c_1,\dots,c_h)}(c_1\Ga_1+\cdots+c_h\Ga_h),\]
  where $(c_1,\dots,c_h)$ runs over a complete set of direction vectors in $\F_q^h$, made specific by having the last non-zero entry equal to 1, i.e.
    \[\det(M(\Ga_1,\dots,\Ga_h))=\prod_{1\le i\le h}\prod_{(c_1,\dots,c_{i-1})}(c_1\Ga_1+\cdots+c_{i-1}\Ga_{i-1}+\Ga_i).\]
  Thus, $\det(M(\Ga_1,\dots,$ $\Ga_h))\neq 0$ if and  only if  $\Ga_1,\dots,\Ga_h$ are $\F_q$-linearly independent.

  \begin{lemma}~\label{lem:2.3} Let $K,L$ be two $(q,\ell)$-Moore matrices of size $h\times n$. Let $A$ be a matrix in $\F_q^{n\times m}$. Then we have
  \begin{itemize}
  \item[{\rm (i)}] $K-L$ is a $(q,\ell)$-Moore matrix of size $h\times n$.
   \item[{\rm (ii)}] $KA$ is a $(q,\ell)$-Moore matrix of size $h\times m$. Furthermore, every element in the first row of $KA$ is an $\F_q$-linear combination of all the elements in the first row of $K$.
  \end{itemize}
  \end{lemma}
  \begin{proof} Let $K,L$ be Moore matrices $M(\Ga_1,\dots,\Ga_n)$ and $M(\Gb_1,\dots,\Gb_n)$ for some $\Ga_i,\Gb_j\in\F_\ell$, respectively. Then it is straightforward to verify that $K-L$ is equal to the Moore matrix $M(\Ga_1-\Gb_1,\dots,\Ga_n-\Gb_n)$. This completes the proof for Part (i).

  Now let $A=(a_{ij})_{1\le i\le n,1\le j\le m}\in \F_q^{n\times m}$. Then an easy computation shows that \[KA=M\left(\sum_{i=1}^na_{i1}\Ga_i,\dots,\sum_{i=1}^na_{im}\Ga_i\right).\]
  Note that we use the fact that $a_{ij}=a_{ij}^q$ in the above identity.
  \end{proof}

  \subsection{Subfield subcodes}
  Let $\ell=q^r$ with $r\ge 1$. For an $\ell$-ary linear code $\mC$ of length $n$, we define a subfield subcode by $\mC|_{\F_q}:=\mC\cap \F_q^n$.
  \begin{lemma}\label{lem:2.4} Let $\ell=q^r$. Then there is a $q$-ary $[n, n-(n-k)r,d]$-linear code provided there exists an $\ell$-ary $[n,k,d]$-linear code. In particular, we have $[n, n-(d-1)r,d]$-linear code for any $n\le q^r+1$.
  \end{lemma}
  \begin{proof}
  Let $\mC$ be an $\ell$-ary $[n,k,d]$-linear code. Then it is clear that  $\mC|_{\F_q}$ is a $q$-ary linear code of length $n$ and minimum distance at least $d$. Furthermore, we have
  \[\dim_{\F_q}(\mC|_{\F_q})= \dim_{\F_q}(\mC)+\dim_{\F_q}(\F_q^n)-\dim_{\F_q}(\mC+\F_q^n)\ge kr+n-nr.\]
  Thus, we proved the first part.

  Now let us consider an $\ell$-ary $[n,n-d+1,d]$-Reed-Solomon code for $n\le q^r+1$. Then, by the first part, we get a $q$-ary $[n, n-(d-1)r,d]$-linear code.
  \end{proof}

  If we consider BCH codes which are also subfield subcodes of Reed-Solomon codes for $q=2$, we can improve the code parameters given in Lemma \ref{lem:2.4}.

    \begin{lemma}{\cite[Proposition 8.1.14]{CT}}\label{lem:2.5} There is a binary $\left[n=2^r-1, n- \Gd r,2\Gd+1\right]$-linear code.
  \end{lemma}

\subsection{Extreme codes}~\label{subsection:2.4}
In this subsection, we consider linear codes of fixed minimum distance with length tending to infinity. In this parameter regime, we wonder how large the dimension of such a code could be. In view of this, we define the following quantity. For a prime power $q$ and an integer $d\ge 2$, define
\begin{equation}\label{eq:10}
c(q,d):=\liminf_{n(\mC)\rightarrow\infty}\frac{n(\mC)-k(\mC)}{\log_qn(\mC)},
\end{equation}
where $n(\mC),k(\mC)$ and $d(\mC)$ stand for length, dimension and minimum distance of $\mC$ and the code $\mC$ runs through all $q$-ary linear codes. By the Hamming bound, we have that, for a $q$-ary $[n,k,d]$-linear code with $d\ge 3$,
\[q^k\le\frac{q^n}{\sum_{i=0}^1{n\choose i}(q-1)^i}\le\frac{q^n}{n(q-1)}.\]
This gives $n-k\ge \log_q n+\log_q(q-1)$. This implies that $c(q,d)\ge 1$ for $d\ge 3$.

One has the following results on $c(q,d)$.
\begin{lemma}\label{lem:2.6} One has
\begin{itemize}
\item[{\rm (i)}] $c(2,d)=\left\lfloor\frac{d-1}2\right\rfloor$ for all $d\ge 3$  and $c(q,6)\le 3$.
\item[{\rm (ii)}] For $d\ge 3$ and ${\rm Char}(\F_q)>d-3$, $c(q,d)\le d-3+\frac1{d-2}$. In particular, we have $c(q,3)=1$, $c(q,4)\le 1.5$ and $c(q,5)\le \frac73$ if $q$ is odd.
\end{itemize}
\end{lemma}
The first result of Part (i) is derived from binary BCH codes and the second result of Part (i) was given in \cite{D95}. Part (ii) was given in \cite{YD04}. Let us cite code parameters given in \cite{D95} and \cite{YD04} below.
\begin{lemma}[\cite{D95,YD04}]\label{lem:2.7} Let $q$ be a prime power. Then
\begin{itemize}
\item[{\rm (i)}] There are $q$-ary $[n=q^{\lfloor5(u-1)/6\rfloor},n-5u/2,6]$-cyclic codes for all $u$ that is divisible by $6$.
\item[{\rm (ii)}]
If ${\rm Char}(\F_q)>d-3$, then there are $q$-ary $[q^u,q^u-(d-3)u-\lceil u/(d-2)\rceil-1,d]$-linear codes for all $u> (d-3)!$.
\end{itemize}
\end{lemma}
Note that in (i) of Lemma \ref{lem:2.7}, we have $\log_qn= \lfloor5(u-1)/6\rfloor= \frac{5u}6-1$, i.e., $u= \frac65(1+\log_qn)$. This implies that the dimension of the code given in (i) of Lemma \ref{lem:2.7} is  $n-\left(3\log_qn+3\right)$.

\subsection{Brief introduction to algebraic geometry codes}
Goppa geometry codes were discovered by Goppa in 1980's \cite{Go81}. Due to their good parameters,  Goppa geometry codes have been extensively studied and applied to various problems. In this subsection, we briefly introduce this class of codes as we only need code parameters of  Goppa geometry codes. For the details, the reader may refer to the book \cite{St09}.

Let $\mX$ be an algebraic curves over $\F_q$ with genus $g$ and $n+1$ pairwise distinct points $\Pin,P_1,\dots,P_n$. Then the Goppa geometry code defined by
\[C(\mP,m\Pin):=\{(f(P_1),\dots,f(P_n)):\; f\in\mL(m\Pin)\}\] is a $q$-ary $[n,m-g+1,n-m]$-linear codes for any $g\le m<n$, where $\mP=\{P_1,P_2,\dots,P_n\}$ and $\mL(m\Pin)$ is the Riemann-Roch space associated with the divisor $m\Pin$.

If $q$ is a perfect square, then, for any $s\ge 1$, the $s$th layer of the Garcia-Stichtenoth tower \cite{GS95a} is an algebraic curves over $\F_q$ with genus $g$ and at least $n+1$ pairwise distinct points, where
\[g\le (\sqrt{q}+1)q^{(s-1)/2};\quad n=(q-1)q^{(s-1)/2}.\]
Thus, we obtain a Goppa geometry codes with length $n=(q-1)q^{(s-1)/2}$, dimension at least $m-(\sqrt{q}+1)q^{(s-1)/2}+1\ge m+1-\frac{n}{ \sqrt{q}-1}$ and minimum distance at least $n-m$ for all $s\ge 1$ and $g\le m<n$.

By Putting $m=n-h-1$ for the above Goppa geometric code, we obtain the following result.
\begin{lemma}\label{lem:2.8} If $q$ is a perfect square, then, for any $u\ge 1$ and $1\le h<n$, there is a $q$-ary $[n,k,d]$-Goppa geometry code with parameters satisfying
\[n=(q-1)q^{(u-1)/2},\quad k\ge n-h-\frac{n}{ \sqrt{q}-1},\quad d\ge h+1.\]
\end{lemma}

\section{{Subspace direct sum system}}
Let us define {subspace direct sum system} first.
\begin{defn}\label{def:2} A {subspace direct sum system} consists of a set of $\F_q$-subspaces $\{V_1,V_2,\dots,V_n\}$ of $\F_q^m$ such that (i) $\dim(V_i)=r$ for every $1\le i\le n$, and (ii) any $h$ subspaces  $\{V_{i_1},V_{i_2},\dots,V_{i_h}\}$ with $1\le i_1<i_2<\cdots<i_h\le n$ form a direct sum.  We denote such a {subspace direct sum system} by $\SP_q(n,m,r,h)$.
\end{defn}
\begin{rmk}
Note that when $r=1$, then a {subspace direct sum system} $\SP_q(n,m,1,h)$ is equivalent to a $q$-ary $[n,k\ge n-m,d\ge h+1]$-linear code.  To see this, we choose a nonzero vector $\bv_i$ (viewed as a column vector) from each subspace $V_i$ and form a matrix $H=(\bv_1,\bv_2,\dots,\bv_n)$. Then the code with $H$ as a parity-check matrix has the desired parameters. Conversely, given a  $q$-ary $[n,k\ge n-m,h+1]$-linear code, we can form one-dimensional spaces $V_i$ spanned by the $i$th column of a parity-check matrix. Thus, we obtain  a {subspace direct sum system} $\SP_q(n,m,1,h)$.
\end{rmk}
In the case of $r=1$, for given $q,n,h$, we want the code dimension $k=n-m$ to be large. In other words, we would like to have small $m$. In view of this fact, we have the following definition.
\begin{defn}\label{def:3} For given $q,n,h,r$, we denote by $m_q(n,r,h)$ the smallest $m$ such that there exists a {subspace direct sum system} $\SP_q(n,m,r,h)$.
\end{defn}

\begin{rmk} For a {subspace direct sum system} $\SP_q(n,m,r,h)$, as the dimension of a direct sum of $h$ subspaces is $hr$, we must have $m\ge hr$. Thus, a trivial lower bound on $m_q(n,r,h)$  is $m_q(n,r,h)\ge hr$.
\end{rmk}
In this section, we are interested in the value $m_q(n,r,h)$.
\subsection{Bounds}
In this subsection, we mainly study upper and lower bounds on $m_q(n,r,h)$. For our constructions of maximally recoverable codes, we are mainly interested in upper bounds on  $m_q(n,r,h)$. However, to see how small the value $m_q(n,r,h)$ could be, we also study some lower bounds. Let us give the standard upper bound, i.e., the Gilbert-Varshamov bound.

\begin{lemma}[Gilbert-Varshamov bound]\label{lem:3.1}
There exists a {subspace direct sum system}  $\SP_q(n,m,r,h)$ with $h\ge 1$ if
\begin{equation}\label{eq:11}
q^m>q^{r-1}\sum_{i=0}^{h-1}{n-1\choose i}(q^r-1)^i.
\end{equation}
\end{lemma}
\begin{proof} We shall prove that, if \eqref{eq:11} holds, then there exists a matrix
\[H=(\bu_{11},\dots,\bu_{1r},\dots,\bu_{n1},\dots,\bu_{nr})\in\F_q^{m\times nr}\]
 such that any $h$ groups $\{\bu_{i_11},\dots,\bu_{i_1r}\},\dots,\{\bu_{i_h1},\dots,\bu_{i_hr}\}$ are linearly independent. Thus, we let $V_i$ be the subspace spanned by $\{\bu_{i1},\dots,\bu_{ir}\}$ and we obtain the desired  {subspace direct sum system}  $\SP_q(n,m,r,h)$.

 First of all, we take $hr$ linearly independent vectors  $\{\bu_{11},\dots,\bu_{1r},\dots,\bu_{h1},\dots,\bu_{hr}\}\subseteq\F_q^m$. For any $h+1\le i\le n$ and $1\le j\le r$, let $\bu_{ij}$ be any vector that is not in the linear span of any $h-1$ groups $\{\bu_{k_11},\dots,\bu_{k_1r}\},\dots,\{\bu_{k_{h-1}1},\dots,\bu_{k_{h-1}r}\}$ ($1\le k_1<\cdots<k_{h-1}\le i-1$) and $\{\bu_{i1},\dots,\bu_{i,j-1}\}$. The number of vectors in the linear span is upper bounded by
 \begin{equation}\label{eq:12}
 \sum_{t=0}^{h-1}{i-1\choose t}(q^r-1)^tq^{j-1}\le  \sum_{t=0}^{h-1}{n-1\choose t}(q^r-1)^tq^{r-1}.
 \end{equation}
Thus, by \eqref{eq:11}, we can always find such a vector $\bu_{ij}$ that is not a linear combination of vectors in these $h-1$ blocks and  $\{\bu_{i1},\dots,\bu_{i,j-1}\}$. Hence, the matrix $H$ is constructed and the desired result follows.
\end{proof}

\begin{cor}\label{cor:3.2} For given prime power $q$ and positive integers $n,h,r$,
there exists a {subspace direct sum system}  $\SP_q(n,m,r,h)$ with
\begin{equation}\label{eq:12}
m:= r+\left\lfloor\log_q\left(\sum_{i=0}^{h-1}{n-1\choose i}(q^r-1)^i\right)\right\rfloor.
\end{equation}
\end{cor}
\begin{proof}
It is easy to verify that inequality \eqref{eq:11} holds for $m=r+\left\lfloor\log_q\left(\sum_{i=0}^{h-1}{n-1\choose i}(q^r-1)^i\right)\right\rfloor$. Hence, by Lemma \ref{lem:3.1}, there exists a  {subspace direct sum system}  $\SP_q(n,m,r,h)$. The desired result follows.
\end{proof}

To derive other upper bounds and Hamming bound, we will establish a connection between  {subspace direct sum systems} and linear block codes.
Let us introduce  block Hamming  weight.
\begin{defn}\label{def:4}  The block Hamming weight $\wt_B(\bc)$ of a vector $\bc=(\bc_1,\bc_2,\dots,\bc_n)\in\F_q^{nr}$ with each $\bc_i\in\F_q^r$ is defined to be $\wt_B(\bc):=|\{1\le i\le n:\; \bc_i\neq\bo\}|$. The block Hamming distance between two vectors $\bu,\bv\in \F_q^{nr}$ is defined to be $d_B(\bu,\bv)=\wt_B(\bu-\bv)$.
\end{defn}

\begin{lemma}\label{lem:3.3} The block Hamming distance $d_B$ defined in Definition \ref{def:4} is indeed a distance.
\end{lemma}
\begin{proof} Let $\phi$ be an $\F_q$-isomorphism between $\F_q^r$ and $\F_{q^r}$. Then the block Hamming weight $\wt_B(\bc)$ is exactly the Hamming weight of   $\phi(\bc)=(\phi(\bc_1),\phi(\bc_2),\dots,\phi(\bc_n))$. Thus, The block Hamming distance in $\F_q^{nr}$ is in fact the Hamming distance in $\F_{q^r}^n$.
\end{proof}

Now we define  block codes (we abuse the notation ``block codes" here as a classical linear code is also called a linear block code).
\begin{defn}\label{def:5}  A linear block code $\mC$ is an $\F_q$-subspace of $\F_q^{nr}$. The minimum distance $d_B(\mC)$  of $\mC$ is defined to be $\min\{d_B(\bu,\bv):\; (\bu,\bv)\in\mC\times\mC,\ \bu\neq\bv\}$. A $q$-ary linear block code $\mC$ in $\F_q^{nr}$ with dimension $k$ and minimum distance $d$ is denoted by $[(n,r),k,d]_q$.
\end{defn}

Next result shows connection between subspace direct sum system  and  linear block codes.
\begin{lemma}\label{lem:3.4}
There exists a {subspace direct sum system}  $\SP_q(n,m,r,h)$ with $h\ge 1$ if and only if there is an $[(n,r),k\ge nr-m,d\ge h+1]_q$-block code.
\end{lemma}
\begin{proof} Suppose that there is {subspace direct sum system}  $\SP_q(n,m,r,h)$ with $h\ge 1$. Let $V_1,\dots,V_n\subseteq\F_q^m$ be such a {subspace direct sum system}. Let $\{\bu_{i1},\bu_{i2},\dots,\bu_{ir}\}$ be an $\F_q$-basis of $V_i$ and define a matrix
\[H=(\bu_{11},\bu_{12},\dots,\bu_{1r},\dots,\bu_{n1},\bu_{n2},\dots,\bu_{nr}).\]
Define the code
\[\mC=\{\bc\in\F_q^{nr}:\; \bc H^T=\b0\}.\]
Then $\mC$ is an $[(n,r),k\ge nr-m]$-block code. As any $h$ subspaces out of $\{V_1,\dots,V_n\}$ are direct sum, the block Hamming weight  of a nonzero of codeword in  $\mC$ is at least $h+1$.

Conversely, if there exists an $[(n,r),k\ge nr-m,d\ge h+1]_q$-block code. Let $\mC^{\perp}$ be the Euclidean dual of $\mC$. Then the generator matrix $H$ is an $(nr-k)\times nr$ matrix. We partition column vectors of $H$ into $n$ groups from left to right, each has $r$ column vectors. As the minimum Hamming weight of $\mC$ is at least $h+1$, any $h$ groups of column vectors of $H$ are linearly independent. Let $V_i$ be the vector space spanned by the $i$th group, then we obtain the desired {subspace direct sum system} $\SP_q(n,nr-k,r,h)$. As $nr-k\le m$, we can embed $\F_q^{nr-k}$ into $\F_q^m$. The desired result follows.
\end{proof}

For a vector $\bc\in\F_q^{nr}$ and an integer $d$ with $0\le d\le n$, define the ball
\[\mB(\bc,d):=\{\bu\in\F_q^{nr}:\; d_B(\bu,\bc)= d\}.\]
Then the size of $|\mB(\bc,d)|$ is independent of the center $\bc$ and given by $\sum_{i=0}^d{n\choose d}(q^r-1)^i$.

\begin{lemma}[Hamming bound]\label{lem:3.5}
If there exists a {subspace direct sum system}  $\SP_q(n,m,r,h)$ with $h\ge 2$, then
\begin{equation}\label{eq:14}
q^m\ge \sum_{i=0}^{\left\lfloor\frac{h}2\right\rfloor}{n\choose i}(q^r-1)^i.
\end{equation}
This implies that
\begin{equation}\label{eq:15}
m_q(n,r,h)\ge \log_q\left(\sum_{i=0}^{\left\lfloor\frac{h}2\right\rfloor}{n\choose i}(q^r-1)^i\right).
\end{equation}
\end{lemma}
\begin{proof} By Lemma \ref{lem:3.4}, there exists an $[(n,r),k\ge nr-m,d=h+1]_q$-block code $\mC=\{\bc_1,\dots,\bc_M\}$ with $M=q^k\ge q^{nr-m}$.  Consider the balls $\mB(\bc_i,\left\lfloor\frac{h}2\right\rfloor)$. As $\mC$ has distance $h+1$, these balls are pairwise disjoint. Since $\bigcup_{i=1}^M\mB(\bc_i,\left\lfloor\frac{h}2\right\rfloor)\subseteq\F_q^{nr}$, we have
\[q^{nr}\ge\sum_{i=1}^M\left|\mB\left(\bc_i,\left\lfloor\frac{h}2\right\rfloor\right)\right|=M\sum_{i=0}^{\left\lfloor\frac{h}2\right\rfloor}{n\choose i}(q^r-1)^i\ge q^{nr-m}\sum_{i=0}^{\left\lfloor\frac{h}2\right\rfloor}{n\choose i}(q^r-1)^i.\]
The desired result follows.
\end{proof}

\begin{lemma}[Singleton bound]
If there exists an $[(n,r),k,d]$-block code, then
\begin{equation}\label{eq:16}
k\le r(n-d+1).
\end{equation}
This implies that
\begin{equation}\label{eq:17}
m_q(n,r,h)\ge hr.
\end{equation}
\end{lemma}
\begin{proof} Deleting the last $d-1$ blocks of length $r$ of $\mC$ to obtain a block code $\mC'$ that is a subset of $\F_q^{r(n-d+1)}$. It is clear that $\mC'$ has the same size as $\mC$ since $\mC$ has the block distance $d$. This gives
\[q^k= |\mC|=|\mC'|\le  q^{r(n-d+1)}.\]
The desired result follows.

Since we have a  {subspace direct sum system}  $\SP_q(n,m,r,h)$ with $m=m_q(n,r,h)$,
by Lemma \ref{lem:3.4}, there exists an $[(n,r),k\ge nr-m,d=h+1]_q$-block code. Thus, we have
\[nr-m\le k\le r(n-h),\]
i.e., $m\ge hr$.
\end{proof}

\subsection{Construction}
In the previous subsection, we studied some bounds on {subspace direct sum systems.}  In particular, we derive an existence result, i.e., the Gilbert-Varshamov bound. In this subsection, we provide some explicit constructions of  {subspace direct sum systems.}

\begin{theorem}\label{thm:3.7} If $1\le t<n\le q^r+1$, then  there exists a $q$-ary  $[(n,r),tr,d= n-t+1]$-block code that achieves the Singleton bound. Furthermore, the code can be constructed explicitly.
\end{theorem}
\begin{proof}
Let $\mC$ be a $q^r$-ary $[n,t,n-t+1]$-MDS code. Let $\pi$ be an $\F_q$-linear isomorphism between $\F_{q^r}$ and $\F_q^r$. Then the code
\[\pi(\mC):=\{(\pi(c_1),\dots,\pi(c_n)):\; (c_1,\dots,c_n)\in\mC\}\]
is an $\F_q$-linear code with length $nr$ and dimension $k=tr$. The minimum block distance of $\pi(\mC)$ is the same as the Hamming distance of $\mC$ which is $n-t+1$. It is straightforward to verify that the code achieves the Singleton bound.

As $\mC$ can be explicitly constructed, the code $\pi(\mC)$ is explicitly constructed as well.
\end{proof}

\begin{cor}\label{cor:3.8} If $n\le 1+{q^r}$, then, for any $1\le h<n$, there exists a {subspace direct sum system} $\SP_q(n,hr,r,h)$ that can be explicitly constructed.
\end{cor}
\begin{proof} Taking $t=n-h$ in Theorem \ref{thm:3.7} gives a $q$-ary $[(n,r),k=r(n-h),d=h+1]$-block code.  The desired result follows from Lemma \ref{lem:3.4}.
\end{proof}
In Corollary \ref{cor:3.8}, the code length $n$ is bounded by $1+q^r$. In order to break this barrier, we consider codes over extension fields and then take subfield subcode.
\begin{theorem}\label{thm:3.9} Let $u\ge 1$ be an integer and let $n=1+q^{ur}$, then for any $2\le h<n$, there exists a {subspace direct sum system} $\SP_q(n,h\log_qn,r,h)$.
\end{theorem}
\begin{proof}
Let  $\ell=q^u$. By Theorem \ref{thm:3.7}, there exists an $\ell$-ary $[(n,r),k=(n-h)r,d\ge h+1]$-block code $\mC$. Consider the intersection $\mC\cap \F_q^{nr}$. Then
\[\dim_{\F_q}(\mC\cap \F_q^{nr})\ge\dim_{\F_q}(\mC)+\dim_{\F_q}(\F_q^{nr})-\dim_{\F_q}(\F_\ell^{nr})=nr-hru\ge nr-h\log_qn. \]
Hence, $\dim_{\F_q}(\mC\cap \F_q^{nr})$ is a $q$-ary $[(n,r), \ge nr-h\log_qn,\ge h+1]$-block code.
The desired result follows from Lemma \ref{lem:3.4}.
\end{proof}
We now make use of the extreme codes given in Subsection~\ref{subsection:2.4} to derive {subspace direct sum systems}.
\begin{theorem}\label{thm:3.10} Let $h$ be a constant.
\begin{itemize}
\item[{\rm (i)}]
If $h\ge 2$ and ${\rm Char}(\F_q)>h-2$.  Then, for any $u\ge (h-2)!$ with $(h-1)|u$, there is a   {subspace direct sum system}  $\SP_q(n,m,r,h)$ with $m=\left(h-2+\frac1{h-1}\right)\log_{q}n+r$ and $n=q^{ur}$.
\item[{\rm (ii)}]  There is a family $\SP_q(n,m,r,5)$ of {subspace direct sum systems}  with $n=q^{\lfloor5(u-1)/6\rfloor}$ and $m=3\log_{q}n+3r$ for all  $u$ that is divisible by $6$.
    \end{itemize}
\end{theorem}
\begin{proof} (i) By Lemma \ref{lem:2.7}, there exists a $\ell=q^r$-ary   $[\ell^u,\ell^u-(h-2)u-\lceil u/(h-1)\rceil-1,h+1]$-linear code $\mC$ for all $u> (d-3)!$. Put $n=\ell^u=q^{ur}$.
Let $\pi$ be an $\F_q$-linear isomorphism between $\F_{q^r}$ and $\F_q^r$. Then the code
\[\pi(\mC):=\{(\pi(c_1),\dots,\pi(c_n)):\; (c_1,\dots,c_n)\in\mC\}\]
is an $\F_q$-linear code with length $nr$ and dimension
\begin{eqnarray*}
k&=&\left(n-(h-2)u-\lceil u/(h-1)\rceil-1\right)r= nr-\left((h-2)-\frac1{h-1}\right)ur-r\\
&=&nr-\left((h-2)-\frac1{h-1}\right)\log_qn-r.\end{eqnarray*}
The block distance of $\pi(\mC)$ is the same as the minimum distance of $\mC$ which is $h+1$. The desired result follows from Lemma \ref{lem:3.4}.

(ii) Let $\ell=q^r$ and $n=\ell^{\lfloor5(u-1)/6\rfloor}$ for an even $u\ge 4$. Consider  a family $\{\mC\}$ of $\ell$-ary $[n,n-3\log_\ell n-\frac52,6]_\ell$-linear codes. Let $\pi$ be an $\F_q$-linear isomorphism between $\F_{q^r}$ and $\F_q^r$. Then the code $\pi(\mC)$
is an $\F_q$-linear code with length $nr$ and dimension $k=(n-3\log_{\ell}n-5/2)r=nr-3\log_{q}n-5r/2$ and minimum distance $6$. Thus, by Lemma \ref{lem:3.4}, we obtain a family $\SP_q(n,m,r,5)$ of {subspace direct sum systems}  with $n\rightarrow\infty$  and $m=3\log_q n+5r/2$.
\end{proof}

Finally we make use of Goppa geometric codes to construct {subspace direct sum systems}.
\begin{theorem}\label{thm:3.11} If $r$ is even or $q$ is a perfect square, then for any $u\ge 1$, there is a family $\SP_q(n,m,r,h)$ of {subspace direct sum systems}  with $m=\left(h+\frac{n}{q^{r/2}-1}\right)r$ and  $n=(q^r-1)q^{r(u-1)/2}$ that can be explicitly constructed.
\end{theorem}
\begin{proof}
Let $\mC$ be a $q^r$-ary $[n,k\ge n-h-\frac{n}{q^{r/2}-1},h+1]$-Goppa geometric code with $n=(q^r-1)q^{r(u-1)/2}$ given in Lemma \ref{lem:2.8}. Let $\pi$ be an $\F_q$-linear isomorphism between $\F_{q^r}$ and $\F_q^r$. Then the code
\[\pi(\mC):=\{(\pi(c_1),\dots,\pi(c_n)):\; (c_1,\dots,c_n)\in\mC\}\]
is an $\F_q$-linear code with length $nr$ and dimension $k\ge \left(n-h-\frac{n}{q^{r/2}-1}\right)r$. The minimum block distance of $\pi(\mC)$ is the same as the Hamming distance of $\mC$ which is at least $h+1$. By Lemma \ref{lem:3.4}, we obtain a subspace direct sum system $\SP_q(n,m,r,h)$.

As $\mC$ can be explicitly constructed, the code $\pi(\mC)$ is explicitly constructed as well. Thus, the desired {subspace direct sum systems} are also explicitly constructed.
\end{proof}

\section{Maximally recoverable codes}
In this section, we present constructions of maximally recoverable codes via {subspace direct sum systems} that have been investigated before. We first give a direct construction of maximally recoverable codes without concatenating with classical linear codes. As locality $r$ is relatively small for most of constructions of {subspace direct sum systems} given above, one needs to concatenate this direct construction of maximally recoverable codes with classical linear codes in order to enlarge locality.
\subsection{Direct construction}~\label{sec:4.1}
In this subsection, we present a direct construction of maximally recoverable codes without concatenating with classical linear codes. Let us state our main results below.
\begin{theorem}\label{thm:4.1}
Let $q$ be a prime power with $q\ge r-1$. If there exists a {subspace direct sum system} $\SP_q(n,m,r,h)$, then there exists an $(N=nr,r,h,{\Gd})_\ell$-MR code with $\ell=q^m$.   In particular,  there exists an $(N=nr,r,h,{\Gd})$-MR code over field $\ell:=q^{m_q(n,r,h)}$. Furthermore, the MR code can be explicitly constructed as long as the {subspace direct sum system} is explicit.
\end{theorem}
\begin{proof}
As $r\le q+1$, there exists a $q$-ary $[r,r-{\Gd},{\Gd}+1]$-MDS code. Let $A$ is a parity-check matrix of such a code. Then $A$ is ${\Gd}\times r$ matrix over $\F_q$ satisfying that every $\delta\times \delta$ submatrix is invertible.

Assume that  $\F_q$-subspaces $\{V_1,\dots,V_n$\} of $\F_{q^m}$  form a {subspace direct sum system} $\SP_q(n,m,r,h)$. Choose an $\F_q$-basis $\{\Ga_{i1},\dots,\Ga_{ir}\}$ of $V_i$ and define the matrix
\begin{equation}\label{eq:18}
D_i=\begin{pmatrix}
\Ga_{i1}&\Ga_{i2}&\dots&\Ga_{ir}\\
\Ga^q_{i1}&\Ga^q_{i2}&\dots&\Ga^q_{ir}\\
\dots&\dots&\dots&\dots\\
\Ga^{q^{h-1}}_{i1}&\Ga^{q^{h-1}}_{i2}&\dots&\Ga^{q^{h-1}}_{ir}
\end{pmatrix}.\end{equation}
Let $H$ be the matrix given in \eqref{eq:8} with $A_i=A$ and $D_i$ defined in \eqref{eq:18} for all $1\le i\le n$. Let $\mC$ be the $\ell$-linear code with $H$ as a parity-check matrix. We claim that $\mC$ is an $(N=nr,r,h,{\Gd})$-MR code over the field $\F_\ell$. To prove this claim, it is sufficient to show that the conditions (a) and (b) in Lemma \ref{lem:2.2} (ii) are satisfied. (a) is satisfied from the construction. It is enough to prove (b) in Lemma \ref {lem:2.2} (ii). More specifically, we show that every $n\delta+h$  columns consisting of any ${\Gd}$ columns in each block $R_i$ for all $1\le i\le n$ and other arbitrary $h$ columns are $\F_\ell$-linearly independent.

Take any ${\Gd}$ columns together with $h_i$ columns in each $R_i$ with $0\le h_i\le h$ and $\sum_{i=1}^nh_i=h$ and re-arrange  columns of $H$  in \eqref{eq:8}, we get a submatrix of $H$ with the following form
\begin{equation}~\label{eq:19} H_1=\left(\begin{array}{c|c|c|c}
C_1,E_1&O&\cdots&O\\ \hline
O&C_2,E_2&\cdots&O \\ \hline
\vdots&\vdots&\ddots&\vdots \\ \hline
O&O&\cdots&C_n,E_n\\ \hline
L_1,K_1&L_2,K_2&\dots&L_n,K_n
\end{array}
\right)\in\F_{\ell}^{(n\delta+h)\times (n\Gd+h)},\end{equation}
where $C_i, E_i, L_i$ and $ K_i$ are ${\Gd}\times {\Gd}$, ${\Gd}\times h_i$, $h\times {\Gd}$ and $h\times h_i$ matrices, respectively.
 Note that the first ${\Gd}$ columns in each block were previously chosen.

 To show that (b) in Lemma \ref{lem:2.2} (ii) is satisfied, it is sufficient to show that $H_1$ is invertible.
Since $A$ is a parity-check matrix of a $q$-ary $[r,r-{\Gd},{\Gd}+1]$-MDS code, $C_i$ is a ${\Gd}\times {\Gd}$ invertible matrix for all $1\le i\le n$. By the elementary  operations, the matrix $H_1$ in \eqref{eq:19} can be rewritten as
\begin{equation}\label{eq:20}
H_2=\left(\begin{array}{c|c|c|c}
C_1,O&O&\cdots&O\\ \hline
O&C_2,O&\cdots&O \\ \hline
\vdots&\vdots&\ddots&\vdots \\ \hline
O&O&\cdots&C_n,O\\ \hline
O,K_1-L_1C_1^{-1}E_1&O,K_2-L_2C_2^{-1}E_2&\dots&O,K_n-L_nC_n^{-1}E_n
\end{array}
\right).
\end{equation}
Thus, to prove that $H_1$ is invertible, it is equivalent to showing that the
the matrix
\begin{equation}~\label{eq:21}
H_3=(K_1-L_1C_1^{-1}E_1|K_2-L_2C_2^{-1}E_2|\dots|K_n-L_nC_n^{-1}E_n)\end{equation} is invertible.
By Subsection~\ref{subsec:2.2}, it is sufficient to show that $H_3$ is a square $(q,\ell)$-Moore matrix whose  elements on the first row are $\F_q$-linearly independent. We note that $K_i,L_i$ are $(q,\ell)$-Moore matrices and $C_i^{-1}E_i$ are matrices over $\F_q$. Thus, by  Lemma~\ref{lem:2.3}, $K_i-L_iC_i^{-1}E_i$ are $(q,\ell)$-Moore matrices. This shows that $H_3$ is also a $(q,\ell)$-Moore  matrix. Next, we will show that the elements on the first row of $H_3$ are $\F_q$-linearly independent.

Define the set $U:=\{i\in[n]:\; h_i\ge1\}$. Then $|U|\le h$ and $H_3$ can be written as $H_3=(K_i-L_iC_i^{-1}E_i)_{i\in U}$.
 Let $S_i\subseteq [n]$ with $|S_i|=h_i$ for all $1\le i\le n$ such that the column $(\Ga_{ij}, \cdots, \Ga_{ij}^{q^{h-1}})^T$ is a column of $K_i$ if and only if the second index $j\in S_i.$ Similarly,
let $T_i\subseteq [n]$ with $|T_i|=\delta$ for all $1\le i\le n$ such that the column $(\Ga_{ij}, \cdots, \Ga_{ij}^{q^{h-1}})^T$ is a column of $L_i$ if and only if the second index $j\in T_i.$  It is clear that the first row of $K_i-L_iC_i^{-1}E_i$ is $(\Ga_{ij}-\Gb_{ij})_{j\in S_i}$, where $\Gb_{ij}$ is an $\F_q$-linear combination of the first rows of $L_i$, i.e. it is an $\F_q$-linear combination of  $\{\Ga_{il}\}_{l\in T_i}$.
Suppose that
$\sum_{i\in U}\sum_{j\in S_i}\lambda_{ij}(\Ga_{ij}-\beta_{ij})=0$ for some $\Gl_{ij}\in\F_q$. As $\sum_{j\in S_i}\lambda_{ij}(\Ga_{ij}-\beta_{ij})$ belongs to $V_i$ and $\sum_{i\in U}V_i$ is a direct sum, we must have $\sum_{j\in S_i}\lambda_{ij}(\Ga_{ij}-\beta_{ij})=0$ for all $i\in U$.

 Thus, we have
\begin{equation}~\label{eq:22}
0=\sum_{j\in S_i}\lambda_{ij}(\Ga_{ij}-\beta_{ij})=\sum_{j\in S_i}\lambda_{ij}\Ga_{ij}-\sum_{j\in S_i}\lambda_{ij}\beta_{ij}.
\end{equation}
This forces that $\Gl_{ij}=0$ for all $j\in S_i$ since (i) $\Gb_{ij}$ is  an $\F_q$-linear combination of  $\{\Ga_{il}\}_{l\in T_i}$; and (ii) elements $\{\Ga_{ij}\}_{j\in S_i}$ and $\{\Ga_{il}\}_{l\in T_i}$ are $\F_q$-linearly independent. Therefore,
we conclude that (b) in Lemma \ref {lem:2.2} (ii) holds. The proof is completed.
\end{proof}

If ${\Gd}=1$ or $r-1$, then the constraint $r\le q+1$ can be removed as there is always an $[r,r-1,2]$-MDS code for any $r\ge 2$.
\begin{cor}\label{cor:4.2} For ${\Gd}=1$ or $r-1$,
there exists always an $(N=nr,h,r,{\Gd})$-MR code  over field $\ell:=q^{m_q(n,r,h)}$.
\end{cor}

Let us now derive $(N=nr,h,r,{\Gd})$-MR codes for fixed $h$ by combining Theorems \ref{thm:4.1} and \ref{thm:3.10}.
\begin{theorem}\label{thm:4.3} Let $h\ge 2$ be a constant.  Then we have an $(N=nr,h,r,{\Gd})_\ell$-MR code that can be explicitly constructed. The  field size $\ell$ satisfies
\begin{itemize}
\item[{\rm (i)}]  For a constant $r$, $\ell=O\left(n^{h-2+\frac1{h-1}}\right)$. In particular, we have $\ell=O\left(n\right)$ for $h=2$; $\ell=O\left(n^{3/2}\right)$ for $h=3$ and  $\ell=O\left(n^{7/3}\right)$ for $h=4$.
\item[{\rm (ii)}] For $r=o(\log n/\log\log n)$, one has $\ell=O\left(n^{h-2+\frac1{h-1}+o(1)}\right)$. In particular, we have $\ell=O\left(n^{1+o(1)}\right)$ for $h=2$; $\ell=O\left(n^{3/2+o(1)}\right)$ for $h=3$ and  $\ell=O\left(n^{7/3+o(1)}\right)$ for $h=4$.
\item[{\rm (iii)}] For $r=o(\log n/\log\log n)$ and $h=5$, $\ell=O(n^{3+o(1)})$. If $r$ is a constant and $h=5$, then $\ell=O(n^{3})$.
\end{itemize}
\end{theorem}
\begin{proof} (i) By Theorem \ref{thm:3.10}, there exists subspace direct sum system $\SP_q(n,m,r,h)$ with $n=q^{ur}$ and $m\le \left({h-2+\frac1{h-1}}\right)\log_qn+2r$ for all positive integer $u$ with $u>(h-2)!$ over the field $\F_q$ with ${\rm Char}(\F_q)>h-2$. Thus, by Theorem \ref{thm:4.1}, we have an $(N=nr,r,h,{\Gd})_\ell$-MR code with
$\ell=q^m\le q^{\left({h-2+\frac1{h-1}}\right)\log_qn+2r}$. As there is always a prime $q$ between $[r,2r)$, we can choose a prime $q$ such that $q$ is a constant and $q\le \max\{2r, 2(h-2)\}$. Thus, we have
\begin{eqnarray}~\label{eq:23}
\nonumber\ell&\le &q^{\left({h-2+\frac1{h-1}}\right)\log_qn+2r}=q^{2r}\times n^{\left({h-2+\frac1{h-1}}\right)}\\
&=& O\left(\max\{2r, 2(h-2)\}^{2r} \times n^{\left({h-2+\frac1{h-1}}\right)}\right)=O\left( n^{{h-2+\frac1{h-1}}}\right).
\end{eqnarray}

(ii) Note that $n=q^{ur}=\Omega(r^r)$, i.e., $r=O(\log n)$. If $r=o(\log n/\log\log n)$, then we have $\max\{2r, 2(h-2)\}^{2r} =n^{o(1)}$. Hence, by ~\eqref{eq:23}, we get that $\ell=O\left( n^{{h-2+\frac1{h-1}+o(1)}}\right)$.

(iii) In the same way, we get the desired result for $h=5$.
\end{proof}

Before further applying Theorem \ref{thm:4.1} and Corollary \ref{cor:4.2} to derive some MR codes, let us label various subspace direct sum systems constructed in previous sections.

{\footnotesize
\begin{center}~\label{table:2}
Table II \\{Constructions of various subspace direct sum systems}\\ \medskip
{\rm
\begin{tabular}{|c|c|c|c|cl}\hline\hline
\multicolumn{5}{|c|}{$\SP_q(n,m,r,h)$}\\ \hline
No. &${m}$ & \multicolumn{2}{c|}{Restrictions} & \multicolumn{1}{c|}{References}  \\ \hline
$1$&$r+\log_q\sum_{i=0}^{h-1}{n-1\choose i}(q^r-1)^i$ &\multicolumn{2}{c|}{--} &\multicolumn{1}{c|}{Lemma \ref{lem:3.1}} \\ \hline
$2$&$hr$ & \multicolumn{2}{c|}{$n\le 1+q^r$}& \multicolumn{1}{c|}{Corollary \ref{cor:3.8}} \\ \hline
$3$&$h\log_q n$ & \multicolumn{2}{c|}{$r\le \log_q n$}  &\multicolumn{1}{c|}{Theorem \ref{thm:3.9}} \\ \hline
{\multirow{2}{*}{$4$}} & {\multirow{2}{*}{$h+\frac{n}{q^{{r}/{2}}-1}$}} & \multicolumn{2}{c|}{$n=(q^r-1)q^{r(u-1)/2}$} &\multicolumn{1}{c|}{\multirow{2}{*}{Theorem \ref{thm:3.11}}}\\
 & & \multicolumn{2}{c|}{for $u\ge 1$ and $q$ is a prefect square} &\multicolumn{1}{c|}{} \\ \hline\hline
\end{tabular}
}
\end{center}
}
``--" in the above table means that there is no restriction.

Now we apply subspace direct sum systems given in Table I to Theorem \ref{thm:4.1} and Corollary \ref{cor:4.2}.

\begin{theorem} We have $(N=nr,r,h,{\Gd})_\ell$-MR codes for the following parameters.
\begin{itemize}
\item[{\rm (i)}] $\ell=O\left((2r)^{hr}n^{h-1}\right)$. Furthermore,  if $rh=o(\log n/\log\log n)$, we have  $\ell=O\left(n^{h-1+o(1)}\right)$.
In particular, (a) $\ell=O\left(n^{1+o(1)}\right)$  if $hr=o(\log n/\log\log n)$ and $h=2$; (b) $\ell=O\left(n^{2+o(1)}\right)$  if $rh=o(\log n/\log\log n)$ and $h=3$.
\item[{\rm (ii)}] If ${\Gd}=1$ or $r-1$,  then $\ell=O\left(2^{rh}n^{h-1}\right)$.
\item[{\rm (iii)}] If $r=\Omega(\log n/\log\log n)$, then $\ell=O((2r)^{hr})$. If $r=\Theta(\log n/\log\log n)$, then $\ell=O(n^{h})$.
\item[{\rm (iv)}] If $r=O(\log n/\log\log n)$,  then $\ell=O(n^{h})$.
\item[{\rm (v)}]  If $r=O(\log n/\log\log n)$, then $\ell=O\left((2r)^{h+\frac{n}{r^{{r}/{2}}-1}}\right)$.
\end{itemize}
\end{theorem}
\begin{proof} (i) By  the subspace direct sum system given in the first row of Table II, we have
\[q^{m_q(n,r,h)}\le q^r \sum_{i=0}^{h-1}{n-1\choose i}(q^r-1)^i=O\left(q^{hr}n^{h-1}\right).\]
By Theorem \ref{thm:4.1}, there exists an $(N=nr,r,h,{\Gd})_\ell$-MR code with $\ell=O\left(q^{hr}n^{h-1}\right)$. Given the constraint that $r\le q+1$ and the fact that there is a prime $q$ in the interval $[r,2r)$, we have $\ell=O\left((2r)^{hr}n^{h-1}\right)$.

If $hr=o(\log n/\log\log n)$, then $(2r)^{hr}=n^{o(1)}$. The desired result follows.

(ii) In this case, we can take $q=2$ and by (i) we have an $(N=nr,r,h,{\Gd})_\ell$-MR code with $\ell=O\left(2^{rh}n^{h-1}\right)$.

(iii) By Theorem \ref{thm:4.1} and  the subspace direct sum system given in the second row of Table II,  there exists an $(N=nr,r,h,{\Gd})_\ell$-MR code with $\ell=q^{hr}=O((2r)^{hr})$. Note that in this case, we have $n\le 1+q^r\le (2r)^r$. Hence, $r=\Omega(\log n/\log\log n)$. If we take $n=q^r$, then we have $\ell=q^{hr}=n^r$. In this case, we have $r=\Theta(\log n/\log\log n)$.

(iv)  By Theorem \ref{thm:4.1} and  the subspace direct sum system given in the third row of Table II,  there exists an $(N=nr,r,h,{\Gd})_\ell$-MR code with $\ell=q^{h\log_qn}=O(n^{h})$. Note that in this case, we have $r\le\log_qn$, i.e., $r\log r\le \log n$. Hence, $r=O(\log n/\log\log n)$.

(v) Put $q=2^{2\lceil (\log r)/2\rceil}$. Then we have $r\le q\le 2r$. By Theorem \ref{thm:4.1} and  the subspace direct sum system given in the last row of Table II,  there exists an $(N=nr,r,h,{\Gd})_\ell$-MR code with
\[\ell=q^{h+\frac{n}{q^{{r}/{2}}-1}}=O\left((2r)^{h+\frac{n}{r^{{r}/{2}}-1}}\right),\]
where we take a prime $q\in[r,2r)$. Note that in this case, we have $n=(q^r-1)q^{r(u-1)/2}$ for $u\ge 1$. Hence, $r=O(\log n/\log\log n)$.
This completes the proof.
\end{proof}

\subsection{Constructions via linear codes}
Theorem \ref{thm:4.1} gives a construction of MR codes with the field size  $\ell=q^{m_q(n,r,h)}$. It is clear that $m_q(n,r,h)$ decreases as $r$ decreases. Thus, the field size $\ell=q^{m_q(n,r,h)}$ becomes smaller as  $r$ decreases. To decrease the field size, our idea is to construct MR codes of block length $r$ from a {subspace direct sum system} $\SP_q(n,m,s,h)$ with $s< r$. The other motivation comes from the fact that in Subsection~\ref{sec:4.1}, most of MR codes has relatively small  block length $r$. Thus, to get larger  block length $r$, we start with a {subspace direct sum system} $\SP_q(n,m,s,h)$ with $s< r$ and then concatenate with a classical code to obtain an MR code of block length $r$.

From our proof of Theorem \ref{thm:4.1}, it is sufficient to make any $h+{\Gd}$ elements in the first row of each $D_i$ $\F_q$-linearly independent. This is equivalent to the fact that first row of $D_i$ is a parity-check matrix of a $[r,r-s,d\ge h+{\Gd}+1]_q$-linear code.
\begin{theorem}~\label{thm:4.5}
Let $q$ be a prime power with $q\ge r-1$. If there exists a {subspace direct sum system} $\SP_q(n,m,s,h)$ and a $q$-ary $[r,r-s,d\ge h+{\Gd}+1]_q$-linear code, then there exists an $(N=nr,r,h,{\Gd})_\ell$-MR code with $\ell=q^m$. In particular, there exists an $(N=nr,r,h,{\Gd})_\ell$-MR code with
$\ell:=q^{m_q(n,h,s)}$ provided that there exists a $q$-ary $[r,r-s,d\ge h+{\Gd}+1]_q$-linear code.
\end{theorem}
\begin{proof}
Note that in the proof of Theorem \ref{thm:4.1}, we only require that the first row of $K_i$ and $L_i$ are $\F_q$-linearly independent. As the number of columns of $(K_i,L_i)$ is at most ${\Gd}+h$. It is sufficient to require that any ${\Gd}+h$ elements of the first row of $D_i$ are $\F_q$-linearly independent.

Now let $V_1,V_2,\dots,V_n$ be a {subspace direct sum system} $\SP_q(n,m,s,h)$ with $\dim_{\F_q}(V_i)=s$. Due to the fact that there exits a $q$-ary $[r,r-s,d\ge h+{\Gd}+1]_q$-linear code, we can choose $\Gb_{i1},\dots,\Gb_{ir}\subseteq V_i\subseteq\F_{q^m}$ such that any ${\Gd}+h$ elements from  $\{\Gb_{i1},\dots,\Gb_{ir}\}$ are $\F_q$-linearly independent. Put
\[D_i=\begin{pmatrix}
\Gb_{i1}&\Gb_{i2}&\dots&\Gb_{ir}\\
\Gb^q_{i1}&\Gb^q_{i2}&\dots&\Gb^q_{ir}\\
\dots&\dots&\dots&\dots\\
\Gb^{q^{h-1}}_{i1}&\Gb^{q^{h-1}}_{i2}&\dots&\Gb^{q^{h-1}}_{ir}
\end{pmatrix}\]
and the desired result follows.
\end{proof}
If we take $\Gd=1$ in Theorem \ref{thm:4.5}, we can remove the constraint that $r\le q+1$. Thus, we obtain the following corollary.
\begin{cor}~\label{cor:4.6}
Let $q$ be a prime power. If there exists a {subspace direct sum system} $\SP_q(n,m,s,h)$ and a $q$-ary $[r,r-s,d\ge h+2]_q$-linear code, then there exists an $(N=nr,r,h,1)_\ell$-MR code with $\ell=q^m$. In particular, there exists an $(N=nr,r,h,1)_\ell$-MR code with
$\ell:=q^{m_q(n,s,h)}$ provided that there exists a $q$-ary $[r,r-s,d\ge h+2]_q$-linear code.
\end{cor}

To obtain MR codes with good parameters from Theorem \ref{thm:4.5} and Corollary \ref{cor:4.6}, we have to choose proper subspace direct sum systems and classical codes to concatenate. 

\begin{theorem}\label{thm:4.7} There exists an $(N=nr,r,h,{\Gd})_\ell$-MR code with the field size $\ell$ satisfying
\begin{itemize}
\item[{\rm (i)}] For $h+\Gd<r$  and $h+\Gd=\Omega(\log n/\log r)$, $\ell=O((2r)^{h(h+\Gd)})$;
\item[{\rm (ii)}] For $h+\Gd<r$, $\ell=O\left((2r)^{h(h+\Gd)}{n\choose h-1}\right)$.
\end{itemize}
\end{theorem}
\begin{proof} Note that in the proof of this theorem, we will replace $r$ by $s$ in the previous constructions of subspace direct sum systems, and then apply to Theorem \ref{thm:4.1}.
(i) Applying the $q$-ary linear code $[r,r-s,h+\Gd+1]$ with $s=h+\Gd$ and the subspace direct sum system given in the first row of Table II to Theorem \ref{thm:4.1}, we obtain a $\ell$-ary $(N=nr,r,h,\Gd)$-MR code with $r\le q+1$ and
\[\ell=q^{hs}=O((2r)^{h(h+\Gd)}).\]
Note that, in the above formula, we use the fact that there is a prime between $r$ and $2r$. As we require that $n\le 1+q^s=1+q^{h+\Gd}$, we get $h+\Gd=\Omega(\log_q n)=\Omega(\log n/\log r)$.

 (ii) Applying the $q$-ary linear code $[r,r-s,h+\Gd+1]$ with $s=h+\Gd$ and the subspace direct sum system given in  the second row of Table II to Theorem \ref{thm:4.1}, we obtain a $\ell$-ary $(N=nr,r,h,\Gd)$-MR code with $r\le q+1$ and
\[\ell=q^{hs}{n-1\choose h-1}=O\left((2r)^{h(h+\Gd)}{n\choose h-1}\right).\]
This completes the proof.
\end{proof}

Let us now consider MR codes for $\Gd=1$.
\begin{theorem}\label{thm:4.8} If $h+1< r$, then we have
 an $(N=nr,r,h,1)_\ell$-MR code with the field size $\ell$ satisfying
 \begin{itemize}
\item[{\rm (i)}] For $h=\Omega(\log n/\log r)$,   $\ell=(r+1)^{h\lceil(h+1)/2\rceil}$;
\item[{\rm (ii)}] $\ell=O\left((r+1)^{h\lceil(h+1)/2\rceil}{n\choose h-1}\right)$.
\end{itemize}
\end{theorem}
\begin{proof} (i)  Consider the binary $[r,r-\lceil(h+1)/2\rceil\log_2(r+1),h+2]$-BCH codes with $r+1$ being a power of $2$. By Corollary~\ref{cor:3.8}, there exists a binary subspace direct sum $\SP_2(n,hs,s,h)$ with $n\le 1+2^s$. Putting $s=\lceil(h+1)/2\rceil\log_2(r+1)$ and applying Corollary \ref{cor:4.6}, we obtain an $(N=nr,r,h,1)_\ell$-MR code with the field size $\ell=2^{hs}=(r+1)^{h\lceil(h+1)/2\rceil}$.
Note that we have $n\le 1+2^s=(r+1)^{\lceil(h+1)/2\rceil}$. This gives $h=\Omega(\log n/\log r)$.

(ii) Consider the binary $[r,r-\lceil(h+1)/2\rceil\log_2(r+1),h+2]$-BCH codes with $r+1$ being a power of $2$ and a binary subspace direct sum $\SP_2(n,m,s,h)$ given in the first row of Table II. Putting $s=\lceil(h+1)/2\rceil\log_2(r+1)$ and applying Corollary \ref{cor:4.6}, we obtain an $(N=nr,r,h,1)_\ell$-MR code with the field size
\[\ell=O\left(2^{hs}{n\choose h-1}\right)=O\left((r+1)^{h\lceil(h+1)/2\rceil}{n\choose h-1}\right).\]
\end{proof}

\end{document}